\documentclass[a4paper,12pt]{article}

\usepackage[table]{xcolor}

\usepackage{amsfonts,
            amsmath,
            amssymb,
            subcaption,
            amsthm,
            nicefrac,
            enumerate,
            comment,
            authblk,
            mathrsfs,
            mathabx,
            calc,
            tikz,
            tcolorbox,
            stfloats,
            mdframed,
            enumitem,
            pgfplots,
            mleftright
            }

\usepackage[labelfont=bf,format=hang]{caption}
\usepackage[sort,compress]{cite}

\mleftright

\pgfplotsset{compat=1.17}

\usepackage[margin=1in]{geometry}

\usetikzlibrary{arrows.meta,shapes,backgrounds}

\tikzset{myarrow/.style={-{Stealth[length=2mm, width=2mm, sep=1pt]}}} 

\def\phase{0pt}

\usepackage[colorlinks=true]{hyperref}

\usepackage[capitalize,sort,compress]{cleveref}

\def\RR{\mathbb{R}}

\def\NN{\mathbb{N}}
\newcommand\gui{\operatorname{GUI}}

\theoremstyle{plain}
\newtheorem{assumption}[figure]{Assumption}

\newtheorem{theorem}[figure]{Theorem}
\newtheorem{axiom}[figure]{Axiom}

\theoremstyle{definition}
\newtheorem*{proposed*}{Proposed Algorithm}
\newtheorem{remark}[figure]{Remark}

\newtheorem*{jlp_alg*}{Main Algorithm}

\newcommand\bref[2]{%
    \begingroup%
    \hyperlink{#1}{#2}%
    \endgroup}

\begin{document}

\title{Object classification in analytical chemistry via data-driven discovery of partial differential equations}

\author[1,2]{J.\ L.\ Padgett}
\author[4]{Y. Geldiyev}
\author[3]{S. Gautam}
\author[3]{W. Peng}
\author[3]{Y. Mechref}
\author[4]{A. Ibraguimov}

\affil[1]{\footnotesize Department of Mathematical Sciences, University of Arkansas, Fayetteville, Arkansas, USA}
\affil[2]{\footnotesize Center for Astrophysics, Space Physics, and Engineering Research, Baylor University, Waco, Texas, USA}
\affil[3]{\footnotesize Department of Chemistry and Biochemistry, Texas Tech University, Lubbock, Texas, USA}
\affil[4]{\footnotesize Department of Mathematics and Statistics, Texas Tech University, Lubbock, Texas, USA}

\date{\today}

\maketitle

\begin{abstract}
Glycans are one of the most widely investigated biomolecules, due to their roles in numerous vital biological processes. However, few system-independent, LC-MS/MS (Liquid chromatography tandem mass spectrometry) based studies have been developed with this particular goal. Standard approaches generally rely on normalized retention times as well as m/z-mass to charge ratios of ion values. Due to these limitations, there is need for quantitative characterization methods which can be used independently of m/z values, thus utilizing only normalized retention times. As such, the primary goal of this article is to construct an LC-MS/MS based classification of the glycans derived from standard glycoproteins and human blood serum using a Glucose Unit Index as the reference frame in the space of compound parameters. For the reference frame, we develop a closed-form analytic formula via the Green's function of a relevant convection-diffusion-absorption equation used to model composite material transport. The aforementioned equation is derived from an Einstein-Brownian motion paradigm, which provides a physical interpretation of the time-dependence at the point of observation for molecular transport in the experiment. The necessary coefficients are determined via a data-driven learning procedure. The methodology is presented in an abstractly and validated via comparison with experimental mass spectrometer data.
\end{abstract}

\tableofcontents

\section{Introduction}\label{intro}

The biological significance of glycans is evident from the numerous studies demonstrating their roles in living systems. These molecules alone, as well as in conjunction with other biomolecules, participate in important biological functions. For example, glycosylation is one of the major post-translational modifications \cite{1,2} and is known to mediate a broad range of biological processes such as cell recognition \cite{3}, cell signaling \cite{3,4}, immune response \cite{5}, and protein stability \cite{6}. Furthermore, aberrations in glycosylation patterns are found to be related to various diseases, including cancers \cite{3,7,8,9,10,11}.

Glycans display high structural complexity owing to the presence of diverse monosaccharide composition, different linkages, and various branching options \cite{22}. Tandem mass spectrometry has emerged as an effective technique for glycan structural studies \cite{1}. This technique in conjunction with liquid chromatography (LC) provides a powerful tool for studying molecular structures \cite{1,12,13,14,15}. Additionally, various derivatization techniques are employed to pursue sensitive and efficient investigation of the glycans. Some of the derivatization reagents include  2-aminobenzamide, procainamide, aminoxyTMT, RapiFluor-MS (RFMS) labeling, and iodomethane permethylation \cite{23}. Our particular method of choice for the current study is {\em permethylation} as it delivers several advantages over other derivatization techniques \cite{12,13,14,15}. In this technique, methyl groups replace the existing hydrogens, oxygen, and nitrogen atoms in a glycan structure. Permethylated glycans have increased hydrophobicity, which makes them ideal for reverse phase separation. It also prevents fucose (sugar) migration \cite{16,17} and sialic acid loss \cite{18}. Also, due to increased positive ion efficiency of glycans, ionization efficiency is improved, thereby enhancing the sensitivity \cite{19}.

Despite the availability of sensitive structural investigation techniques, inter-in\-stru\-ment, as well as inter-laboratory variations, in the data acquisition complicates the identification and characterization of glycans. This motivates the need for the development of universally applicable instrument, as well as laboratory, independent classification techniques. The Glucose Unit Index (GUI) is one such method, as it relies only on the relative retention time of sample molecules with respect to the glucose units. Ashwood et al.\ recently reported the retention time normalization of native glycans based on GUI as well as m/z  values \cite{20}. 

In this study, we develop a method which only utilizes GUI for characterizing permethylated glycans, independent of m/z values of the permethylated glycan structures. Dextrin, which is utilized as a reference frame in this study, consists of a mixture of oligosaccharides of D-glucose units which form linear chains consisting of either $\alpha$-$(1 \to 4)$ or $\alpha$-$(1\to 6)$ glycosidic bonds. The retention times of these glucose units are used to calculate the normalized retention times of the reduced and permethylated N-glycans derived from samples. The use of Dextrin as an internal standard allows for the elimination of inter-injection variations and improves the accuracy of the measurements. This approach then allows for the development of a mathematical model employing the LC-MS/MS-based data for efficient identification of permethylated N-glycans.

In particular, we employ a so-called data-driven methodology for constructing an associated partial differential equation (PDE) which allows for a straightforward classification procedure. The use of data-driven PDEs and other data-driven methodologies have recently garnered much attention in the literature due to their ability to efficiently learn in relation to dynamical systems and physical processes (see, for instance, \cite{bar2019learning,narasingam2018data,%
flandrin2004empirical,li2020robust,%
brunton2016discovering,rudy2017data,xu2019dl,%
raissi2017physics,xiong2019data,raissi2018hidden,%
berg2019data,schaeffer2017learning,rudy2019data,%
raissi2017machine,long2018pde} and the references therein). The coupling of such approaches with deep learning methods has allowed for the efficient and accurate handling of situations involving quite large data sets. Herein, we consider a modified approach which avoids standard regression methods in favor of a more mathematically informed process for determining the coefficients necessary for object classification. By modifying an approach employed by Einstein (e.g., \cite{einstein}) we are able to deduce more clearly the form of the undetermined PDE --- making our approach closer to a {\em supervised learning} method (with some distinct differences). Moreover, we demonstrate that performing a single learning procedure on a particular data set will allow for highly accurate classification of unknown data sets of a particular type. This, of course, motivates a wide array of novel questions related to learning procedures in both mathematics and the physical sciences.

This article is organized as follows.  \cref{sec:proposed} provides a heuristic description of the algorithm in order to provide the reader clarity regarding our goal and general methodologies. In \cref{einstein-paradigm}, we use a modification of the original Einstein argument for the classical development of standard Brownian motions (see \cite{einstein}, for example) to derive the primary equation of interest. This section also includes a more generalized procedure for the derivation which reduces the necessary assumptions on the system of interest. \cref{closed_form} builds on the material in \cref{einstein-paradigm}, allowing for the construction of a closed-form solution to our theoretical model --- thus, bypassing the need for numerical approximations and complicated learning procedures. We then use this model in \cref{classification_sec} to classify unknown samples via our proposed algorithm. This section also clearly outlines the parameters for the physical experiments carried out to produce the data set for classification. Finally, \cref{conclusions} provides some concluding remarks and also alludes to possible future endeavors related to the current work.


\section{Outline of proposed algorithm}\label{sec:proposed}

We briefly outline the ideas behind the newly proposed algorithm, below. Note that the algorithm will be more completely and rigorously described in \bref{alg1}{Main Algorithm} (see \cref{classification_sec}). For clarity, we allude to specific aspects of the experiment of interest, whose specific protocol is outlined in \cref{GUI-structure}.

\begin{proposed*}\hypertarget{proposed_alg}{}
Let $A$ denote the sample object of interest (e.g., a standard N-glycan --- cf.\ \cref{GUI-structure}).
\begin{enumerate}[label=(\roman *)]
\item 
\label{proposed_alg_1}
Inject into the sample, $A$, simple chemicals which will serve as ``markers'' in the classification process. For our purposes, we use glucose molecules of different types and denote these types by $M_i$, $i\in N^* = \{1, 2, \ldots, n^*\}$ (where $n^* \in \NN$). Each molecule $M_i$, $i\in N^*$, has a linear structure and (possibly) different lengths.
\item
\label{proposed_alg_2}
We then ``slowly'' transport sample $A$ through a short (approximately 10 cm) porous tube. We assume that the transport is one-dimensional and let this transport coincide with the positive $x$-axis (i.e., as in \cref{fig:Scheme of sensing}). (Note that the speed of transport was selected through auxiliary experiments in order to  maximize the high resolution of the signal-to-noise ratio, prior to our classification procedure.)
\item
\label{proposed_alg_3}
At the point of observation, which we denote by $x=L$ ($L < 10$ cm), we record all signals which are obtained from the mass spectrometer (with a particular emphasis on the observed peaks in the signals).  
\item
\label{proposed_alg_4}
In this situation it is assumed that for  all $i,j\in N^*$, such that $i\neq j$, it holds that $M_i$ and $M_j$ do not mix (i.e., do not undergo chemical bonding). 
We identify each $M_i$, $i\in N^*$, via the peaks in the signals of sample $A$ obtained from spectrometer. These  $M_i$, $i\in N^*$, will serve as our ``marker'' molecules. 

Note that the family of ``marker'' molecules $M_i$, $i\in N^*$, will be a subset of all molecules in any other sample of interest: $M_j$, $j\in N = \{1,2,\dots,n\}$ (where $n^* \le n \in \NN$). In other words:
\begin{equation}
\text {The set} \ (M_j)_{j\in N^*} \  \text{is a subset of} \ (M_j)_{j\in N}.
\end{equation}
\item
\label{proposed_alg_5}
We then classify the peaks that are located between those of the ``marker'' molecules in the sample $A$ via a so-called {\em Classification Index}. In the current article, this index  will depend only on two parameters --- which in turn depend only on the retrieval time (the time at which the signal has the largest peak). 
\item
\label{proposed_alg_6}
Using the information obtained by completing
\cref{proposed_alg_1,proposed_alg_2,proposed_alg_3,%
proposed_alg_4,proposed_alg_5}, we may then classify other samples of interest.
That is, for all other samples which have been injected with the same ``markers,'' we extract data regarding the $M_i$, $i\in N^*$, in these samples by matching the spectrometer signal peaks which are closest to the ``marker'' molecules in sample $A$.
\item
\label{proposed_alg_7}
We can then classify the remaining signal peaks in these  samples through associated so-called {\em data-driven PDEs}. This is accomplished by constructing appropriate diffusion and absorption coefficients, which will allow us to distinguish differences between the new samples and the original sample $A$ (see \bref{alg1}{Main Algorithm}). 
(Note that we use the terms {\em diffusiont} and {\em absorption} to mirror the description given in the thought experiment employing compound transport based on Einstein's paradigm of Brownian motion with absorption and drift --- see \cref{einstein-paradigm} for more details.) 
\end{enumerate}
\end{proposed*}

\noindent
The generic description provided by 
\cref{proposed_alg_1,proposed_alg_2,proposed_alg_3,%
proposed_alg_4,proposed_alg_5,proposed_alg_6,proposed_alg_7} in \bref{proposed_alg}{Proposed Algorithm} are meant to provide a rough blueprint of the method we employ, herein. However, there are numerous mathematical details needed before we can rigorously formulate the final algorithm.
A schematic depiction of \bref{proposed_alg}{Proposed Algorithm} (and \bref{alg1}{Main Algorithm}) is provided in \cref{fig:Scheme of sensing}, below.


\begin{figure}[htb!]
\centering
\subcaptionbox{Proposed experiment schematic.\label{f_1_a}}[.6\textwidth]{
\centering
\begin{tikzpicture}
\node at (-1.1,3) [left]{Glucose} ;
\draw[myarrow,ultra thick] (-1,3) -- (-0.2,3) ;
\node at (3,3.3) [above]{Experimental Tube};
\node[regular polygon,regular polygon sides=3,draw,fill] at (5.5,2.65) [below]{};
\node at (5.5,2) [below]{Sensor};
\draw[orange,double=orange!40,double distance=14pt] (0,3) |- (6,3);
\draw[line width=5pt,dashed,blue!80!black,dash pattern={on 12pt off 6pt},dash phase=\phase] (5,3) |- (5.5,3);
\draw[line width=8pt,dashed,blue!80!black,dash pattern={on 12pt off 6pt},dash phase=\phase] (5.5,3) |- (6,3);
\end{tikzpicture}
\vspace{10mm}
}
\subcaptionbox{Data from sensor.\label{f_1_b}}[.35\textwidth]{
\centering
\begin{tikzpicture}[framed,x=0.5cm,y=0.5cm]
  \draw[-stealth] (-0.5,0)--(8,0) node[right]{$t$}; 
  \draw[-stealth] (0,-0.5)--(0,5) node[right]{$C(L,t)$}; 

  \fill[red] (0,0)  circle[radius=2pt];
  \draw[red] (0,0)--(0.5,0) ;
  \fill[red] (0.5,0) circle[radius=2pt];
  \draw[red] (0.5,0)--(1,0) ;
  \fill[red] (1,0) circle[radius=2pt];
  \draw[red] (1,0)--(1.5,0) ;
  \fill[red] (1.5,0) circle[radius=2pt];
  \draw[red] (1.5,0)--(2,0) ;
  \fill[red] (2,0) circle[radius=2pt];
  \draw[red] (2,0)--(2.5,0) ;
  \fill[red] (2.5,0) circle[radius=2pt];
  \draw[red] (2.5,0)--(3,0) ;
  \fill[red] (3,0) circle[radius=2pt];
  \draw[red] (3,0)--(3.5,0) ;
  \fill[red] (3.5,0) circle[radius=2pt];
  \draw[red] (3.5,0)--(4,0) ;
  \fill[red] (4,0) circle[radius=2pt];
  \draw[red] (4,0)--(4.5,0) ;
  \fill[red] (4.5,0) circle[radius=2pt];
  \draw[red] (4.5,0)--(5,2) ;
  \fill[red] (5,2) circle[radius=2pt];
  \draw[red] (5,2)--(5.25,1.5) ;
  \draw[red] (5.25,1.5)--(5.5,5) ;
  \fill[red] (5.5,5) circle[radius=2pt];
  \draw[red] (5.5,5)--(5.75,1.5) ;
  \draw[red] (5.75,1.5)--(6,2) ;
  \fill[red] (6,2) circle[radius=2pt];
  \draw[red] (6,2)--(6.5,0) ;
  \fill[red] (6.5,0) circle[radius=2pt];
  \draw[red] (6.5,0)--(7,0) ;
  \fill[red] (7,0) circle[radius=2pt];
  
  \draw (5.5,0.2)--(5.5,-0.2) node [below]{$T_{\max}$};
\end{tikzpicture}
}
\caption{Schematic diagram of the experiment, with a signal peak and retention time, $T_{\max}$. Note that \cref{f_1_a} is a rough depiction of a generic sample passing through a tube and past a mass spectrometer sensor. Possible raw data obtained from this sensor is provided in \cref{f_1_b}. The value of $T_{\max}$ is determined by direct observation of the data obtained from the sensor.
\label{fig:Scheme of sensing}}
\end{figure}
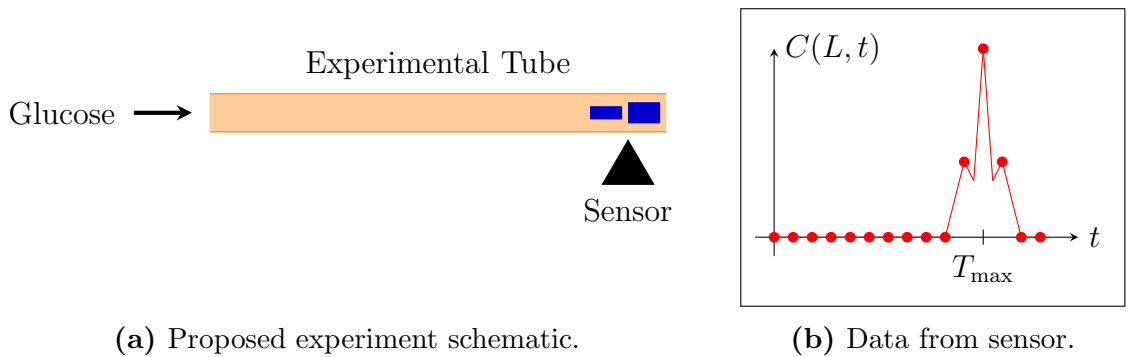

As noted in \cref{proposed_alg_4} of \bref{proposed_alg}{Proposed Algorithm},
our method is based on the important assumption that for all $i,j\in N^*$, such that $i\neq j$, the molecules $M_i$ and $M_j$ are not mixing throughout the experiment.
This is formalized in the following assumption.

\begin{assumption} \label{non-mix}
For each $i,j\in N^*$ (cf.\ \cref{proposed_alg_1} of \bref{proposed_alg}{Proposed Algorithm}), such that $i\neq j$, it is assumed that $M_i$ and $M_j$ are not mixing.
That is, they do not interact to create novel chemical compounds.
\end{assumption}

\cref{non-mix} is a vital assumption in our proposed method. However, it is worth noting that \cref{non-mix} may not be valid in all physical experiments of interest. For our particular situation, empirical evidence (obtained from experiments employing the guidelines outlined in \cref{GUI-structure}) suggests that \cref{non-mix} does in fact hold. The process of molecular mixing requires further generalizations of Einstein's paradigm for Brownian motion and will be a focus of forthcoming research. 
More details on this process are provided in \cref{einstein-paradigm}.

\section{Einstein's paradigm for molecular transport in a tube} \label{einstein-paradigm}

In this section we will reformulate Einstein's model for Brownian motion to the case where glucose molecules are being transported in the tube filled with porous material. Einstein derived his seminal mathematical framework based upon his visual observations of the random jumps of pollen grains of the plant \textit{Clarkia pulchella} suspended in water (see \cite{einstein,Ibragim1}). He then provided a mathematical framework to describe the observed phenomena, which resulted in his model of classical Brownian motion. 
A less understood fact is that Einstein's approach can also be applied to many generic processes arising in physics, chemistry and engineering. This is the key observation employed to derive our novel classification algorithm. 

\subsection{Assumptions for Einstein's paradigm}

In order to employ Einstein's approach in novel situations, one must first understand the key principles which underpinned his work. 
These principles can be formulated into three main axioms, which we formulate into additional assumptions.

\begin{assumption}\label{assump2}
For each molecule of interest, say $M$, we assume that there exists a time interval $\tau$, which is small compared to the observable time intervals but large enough that the motions performed by $M$ during two consecutive time intervals of length $\tau$ can be considered as mutually independent events. 
\end{assumption}

\begin{assumption}\label{assump3}
Let $\cal M(\tau)$ to be the set of all possible lengths of non-colliding jumps associated to $M$ in the time interval $\tau$ (cf.~\cref{assump2}). We will say that $\Delta$ is a possible length of such ``free jumps'' if $\Delta \in  \cal M(\tau)$. 
\end{assumption}

\begin{assumption}\label{assump4}
All molecular interactions during such time intervals 
$\tau$ (cf.~\cref{assump2}) are restricted to absorption of the surrounding media --- which includes other molecules, the porous media, and all possible boundaries. 
\end{assumption}

\begin{remark}\label{remark1}
In general, the time interval of free jumps, $\tau$ (cf.~\cref{assump2}), the expected value of the length of the free jumps associated to $\tau$, which we denote by $\Delta_e \in \mathcal{M}(\tau)$, and the frequency of the free jumps of length $\Delta \in \mathcal{M}(\tau)$, which we denote by $\varphi(\Delta) \in [0,\infty)$, are key parameters in our approach and may depend on underlying properties of the involved molecules (such as rates of change) and the surrounding environment (such as the components of the mixtures and the associated porous media) (see, for example, \cite{Ibragim1} and the references therein).
\end{remark}

\noindent
We will use \cref{remark1} for the interpretation of the experimental data, herein.

\subsection{Derivation of associated equation with drift and absorption}\label{derivation1}

For the development of our PDE model, we employ the notation used in the original work of Einstein \cite{einstein}. 
Denote by $n \in \NN$ the total number of particles of a molecule $M$, presented in a unit volume, and let $\Delta \in \cal M(\tau)$ (cf.~\cref{assump3}). Let $dn \in \NN$ denote the number of particles located in an arbitrary interval of length $d\Delta$ experiencing a displacement of magnitude $\Delta$ in the time interval $\tau$ (cf.~\cref{assump2}).
This value can be expressed via the following equation
\begin{equation}\label{dn}
dn = n \varphi(\Delta) \, d \Delta  ,
\end{equation}
where $\varphi \colon \RR \to [0,1]$ is the frequency (or probability density) of free jumps with length $\Delta.$ 
Note that in his original work, Einstein assumed that $\varphi$ was an even function which differs from zero only for very small values of $\Delta$ (e.g., \cite[Page 13]{einstein}). 
\begin{remark}\label{rem6}
It is worth noting that Einstein's assumptions on $\varphi$ are natural assumptions in the particular case involving Brownian motion. However, it is important to observe that this assumption may not be reasonable for all physical phenomena.
\end{remark}

Based on \cref{rem6}, it should be clear that our intentions are to consider a more general case than that of Einstein's original work. Therefore, we must clearly indicate what assumptions we will employ (as an arbitrary probability density will be too general to allow for any meaningful data-driven learning).
To this end, we assume that there exists $\sigma \in [0,\infty)$ such that the function $\varphi$ satisfies
\begin{equation}\label{prob1}
\int_{-\infty}^{\infty} \varphi(\Delta) \, d\Delta = 1,
\end{equation}
which we refer to as the {\em whole universe axiom}, 
\begin{equation}\label{prob2}
\int_{-\infty}^{\infty} \Delta \varphi(\Delta) \, d\Delta = \Delta_e
\end{equation}
(cf.~\cref{remark1}),
which we refer to as the expected length of the free jumps, and
\begin{equation}\label{prob3}
\int_{-\infty}^{\infty} (\Delta-\Delta_e)^2 \varphi(\Delta) \, d\Delta = \sigma^2 
\end{equation}
(cf.~\cref{remark1}),
which we refer to as the {\em standard variance of the free jump lengths}.

Next, let $f \colon \RR \times [0,\infty) \to [0,\infty)$ be the function which represent the number of particles per unit volume present at position $x \in \RR$ at time $t \in [0,\infty)$. 
With this \cref{prob1}, \cref{prob2}, and \cref{prob3}, we may now formulate a crucial axiomatic conservation law.

\begin{axiom}\label{Einstein_conserv}
There exists a continuous function $F \colon \RR \to \RR$ such that the number of particles found at time $t + \tau$ (cf.~\cref{assump2}) between two planes perpendicular to the $x$-axis with abscissas $x$ and $x + \delta$ (where $\delta \in \RR$) is given by
\begin{equation}\label{Einstein_conserv_eq}
\int_x^{x+\delta} f(y, t+\tau)  \,dy =  \int_x^{x+\delta}\int_{-\infty}^{\infty} f(y + \Delta, t) \varphi(\Delta)\, d \Delta \, dy  +  \int_x^{x+\delta} \int_t^{t+\tau}F\bigl(f(y, s)\bigr)\,ds \, dy.
\end{equation} 
\end{axiom}

The second term in the right-hand-side of \cref{Einstein_conserv_eq} is a result of possible bonding and/or absorption with other molecules in the sample or with the porous media within the tube. In general, for \cref{Einstein_conserv_eq} to be well-defined, all one needs is that $F$ is finite and measurable on the codomain of the function $f$.
However, throughout this article, we will assume that there exists $\omega\in\RR$ such that for all $x\in\RR$, $t \in [0,\infty)$  this term is well-approximated by the linear function  $\omega \cdot f(x,t)$.
The conservation law given by \cref{Einstein_conserv_eq} is depicted schematically in \cref{fig:Schematic}, below.


\begin{figure}[!htb]
\centering
\begin{tikzpicture}
\draw[-,very thick] (-5,0) -- (5,0) ;
\draw[-,very thick] (-4,3) -- (4,3) ;
\filldraw[red] (0,3) circle (2pt) ;
\node[above] at (0,3) {$(x,t+\tau)$} ;
\draw[-,thick] (0,0.1) -- (0,-0.1) ;
\node[below] at (0,-0.1) {$(x,t)$} ;
\draw[-,thick] (3,0.1) -- (3,-0.1) ;
\node[below] at (3,-0.1) {$(x+\delta,t)$} ;
\draw[-,thick] (-3,0.1) -- (-3,-0.1) ;
\node[below] at (-3,-0.1) {$(x-\delta,t)$} ;
\node[right] at (-5,5) {Number of particles here is $f(x,t+\tau)$} ;
\draw[myarrow,ultra thick] (-0.5,4.7) [out=345, in=30] to (0.9,3.4) ;
\draw[myarrow,very thick,blue] (-3,1.8) -- (-0.3,2.8) ;
\filldraw[blue] (-3,1.8) circle (2pt) ;
\draw[myarrow,very thick,blue] (-2.5,0.5) -- (-0.1,2.7) ;
\filldraw[blue] (-2.5,0.5) circle (2pt) ;
\draw[myarrow,very thick,blue] (-0.3,1) -- (0,2.6) ;
\filldraw[blue] (-0.3,1) circle (2pt) ;
\draw[myarrow,very thick,blue] (1,0.8) -- (0.2,2.7) ;
\filldraw[blue] (1,0.8) circle (2pt) ;
\draw[myarrow,very thick,blue] (3,0.25) -- (0.3,2.8) ;
\filldraw[blue] (3,0.25) circle (2pt) ;
\end{tikzpicture}
\caption{Schematic representation of Einstein's Conservation Law given by \cref{Einstein_conserv_eq}. The parameter $\delta$ serves as a boundary for the spatial points which can directly influence the number of particles at $(x,t+\tau)$. \label{fig:Schematic}}
\end{figure}
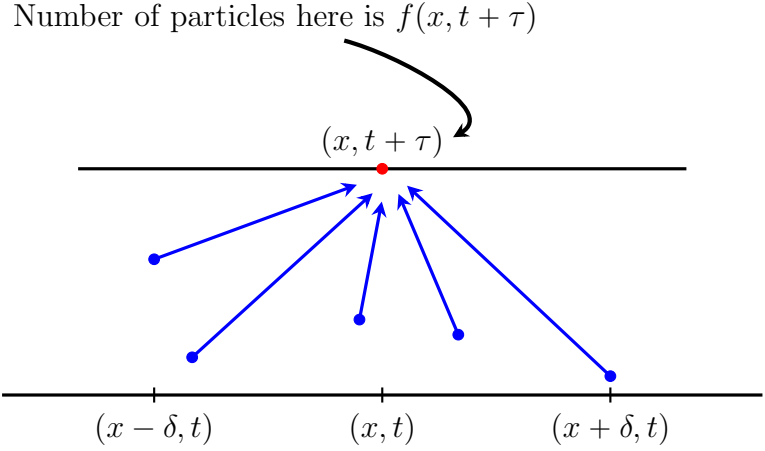

\begin{remark}
In \cref{+-Deltae} we assume that $\varphi$ has compact support. This assumption is physical in nature, as the ensuing theoretical work follows in a similar fashion if $\varphi$ does not have compact support.
\end{remark}

Note that for all $x \in \RR$ it holds that
\begin{equation}\label{+-Deltae}
\int_{-\infty}^{\infty} f(x+ \Delta, t) \varphi(\Delta) \, d \Delta
= \int_{-\infty}^{\infty} \bigl[f(x+ \Delta, t)-f(x+\Delta_e,t) +f(x+\Delta_e,t) \bigr]\varphi(\Delta) \, d \Delta
\end{equation}
(cf.~\cref{remark1}).
Next, observe that the multi-dimensional Carath{\'e}odory Theorem (e.g., Bartle et al.\ \cite[Theorem 6.1.5]{Bartle}) ensures that there exist measurable functions $\psi^x, \psi^t \colon \RR \times [0,\infty) \to \RR$ such that for all $x \in \RR$, $t\in[0,\infty)$ it holds that
\begin{equation}\label{Caratheodory}
f(x, t+ \tau) - f(x+\Delta_e, t) = \tau \bigl[ \psi^t(x,t+\tau) \bigr] + \Delta_e \bigl[ \psi^x(x+\Delta_e,t) \bigr] ,
\end{equation}
where for all $x\in\RR$, $t\in[0,\infty)$ it holds that
\begin{equation}
 \psi^t(x,t+\tau)\approx \frac{\partial f(x,t)}{\partial t}
\qquad \text{and} \qquad 
 \psi^x(x+\Delta_e,t)\approx \frac{\partial f(x,t)}{\partial x}.
\end{equation}
Furthermore, applying Taylor's theorem to $f(x+\Delta, t)$, with respect to $x$, centered at the point $x+\Delta_e$, yields (under appropriate smoothness assumptions) that for all $x\in\RR$, $t\in[0,\infty)$ it holds that
\begin{equation}\label{Taylor}
\begin{split}
 f(x+\Delta, t) & = f(x + \Delta_e, t) + \sum_{k=1}^\infty \frac{(\Delta - \Delta_e)^k}{k!} \frac{\partial^k f}{\partial x^k}(x + \Delta_e,t) \\
& = f(x+\Delta_e, t)+ (\Delta-\Delta_e) \frac{\partial f}{\partial x}(x + \Delta_e,t) + \frac{(\Delta-\Delta_e)^2}{2!} \frac{\partial^{2} f}{\partial x^{2}}(x + \Delta_e,t) \\
& \qquad + \frac{(\Delta-\Delta_e)^3}{3!} \frac{\partial^{3} f}{\partial x^{3}}(x + \Delta_e,t) + \mathcal{O}\bigl( \lvert \Delta - \Delta_e \rvert^4 \bigr).
\end{split}
\end{equation}
Next, we assume that $\lvert \Delta - \Delta_e \rvert \ll 1$ and that for all $k \in \NN$ it holds that
\begin{equation}\label{eq:11}
f(x+\Delta_e,t) \approx f(x,t) \qquad \text{and} \qquad \frac{\partial^k f}{\partial x^k} (x+\Delta_e,t) \approx \frac{\partial^k f}{\partial x^k}(x,t).
\end{equation}
Combining this with \cref{prob1,prob2,prob3,+-Deltae,Taylor} (after disregarding the higher-order terms in \cref{Taylor} --- justified by the assumption that $\lvert \Delta - \Delta_e \rvert \ll 1$), and straightforward calculus demonstrates that for all $x\in\RR$, $t\in[0,\infty)$ it holds that the function $f$ satisfies 
\begin{equation}\label{eq:0}
\tau \frac{\partial f}{\partial t}(x,t) +  \Delta_e\frac{\partial f}{\partial x}(x,t) = \frac{\sigma^2}{2}\frac{\partial^{2} f}{\partial x^{2}}(x,t) +\tau \omega f(x,t).
\end{equation}
For convenience we will rewrite \cref{eq:0} in the form:
\begin{equation} \label{eq:1}
\frac{\partial f}{\partial t}(x,t) - D  \frac{\partial^{2} f}{\partial x^{2}}(x,t) +\gamma \frac{\partial f}{\partial x}(x,t) - \omega f(x,t) = 0,
\end{equation}
where for all $x\in\RR$, $t\in[0,\infty)$, $\tau \in \mathcal{M}(\tau)$ (cf.~\cref{assump2}) the diffusion, drift (convection), and absorption terms are given by
\begin{equation}\label{D-def}
\frac{\sigma^2}{2\tau}= D, \qquad \frac{\Delta_e}{\tau} = \gamma, \qquad\text{and}\qquad
\frac{1}{\tau}\int_t^{t+\tau}F\bigl(f(x, s)\bigr)\,ds  \approx \omega f(x,t),
\end{equation}
respectively.
The differential equation given in \cref{eq:1} is the well-known convection-diffusion-absorption equation which arises in the study of numerous physically relevant phenomena.
It is worth noting that \cref{eq:1} immediately follows from Einstein's original equation for Brownian motion (e.g., \cite[Equation (10)]{einstein}) if we set $\Delta_e=0$ (no drift) and $\omega =0$ (no absorption).

We conclude this subsection with some final remarks.
In the arguments above, we have assumed sufficient smoothness conditions in order to allow for all claims to hold {\em globally} (i.e., for all $t\in[0,\infty)$). However, this can easily be circumvented by constructing \cref{eq:1} {\em locally} (e.g., for some $c \in (0,\infty)$ with $t \in [0,c)$) and then ``gluing'' the results together. This approach is avoided, herein, for simplicity and ease of exposition.
Furthermore, the assumption that $\lvert \Delta - \Delta_e \rvert \ll 1$ is not an inherently restrictive assumption. This assumption loosely can be thought of as removing the possibility of so-called {\em long-range interactions} within the experiment. One can allow for such interactions, but the resulting PDE may involve non-local operators such as the fractional Laplacian (see, for instance, \cite{josh1,josh2,josh3,josh4} and the references therein).
Finally, we note that the assumptions in \cref{eq:11} are purely for convenience, as we can obtain a result similar to \cref{eq:1} through a simple re-scaling of the function $f$.

\begin{remark}
For the remainder of our study, we will assume that $\Delta$ , $\varphi$, and $\tau$ (cf.~\cref{assump2}) are independent of the function $f$. 
In general, these parameters can upon both the spatial and temporal variables, $x$ and $t$, respectively, as well as the underlying porous media. 
Moreover, in more general situations, these parameters can further depend on the dependent variables' derivatives.    
\end{remark} 

\begin{remark}
Herein, we assume that $\tau$ (cf.~\cref{assump2}) is the same for each of the types of molecules of interest. 
That is, for each $M_i$, $i\in\{1,2,\dots,n\}$, we assume that the associated $\tau_i = \tau$.
One can always find such a $\tau$ by simply choosing the minimum of the set of associated $\tau_i$, $i \in \{1,2,\dots,n\}$.
\end{remark}

\subsection{A remark on an alternative derivation}

We conclude \cref{einstein-paradigm} with an alternative derivation of \cref{eq:1}. This alternative derivation is not simply an exercise in pure mathematics, but rather, it provides justification that the proposed algorithm may be applied to a wider class of problems than originally expected. First, this alternative derivation helps to resolve a slight inconsistency in the derivation presented in \cref{derivation1}. That is, in \cref{derivation1} we simultaneously applied the Carath{\'e}odory theorem and Taylor's theorem to \cref{Einstein_conserv_eq}. There is nothing inherently wrong with this approach --- it simply seems strange. The approach presented below rectifies this concern. Next, this alternative derivation reduces the smoothness assumptions one imposes on the function $f$. In practice, it is often assumed that a function of interest is smooth enough to manipulate via expansions, but this assumption excludes numerous physically relevant situations. As such, we present a method for reducing the regularity assumptions needed to obtain \cref{eq:1}, which, in turn, increases the applicability of the proposed method.

To that end, we will derive \cref{eq:1} using only the Carath{\'e}odory theorem. 
Throughout this argument, we assume that for all $x \in \RR$, $t\in[0,\infty)$ it holds that $\frac{\partial^2 f}{\partial x^2}(x,t)$ is continuous and bounded (otherwise the strong form of the PDE \cref{eq:1} is not well-defined).
This implies that there exists $\alpha \in (0,2]$ such that for all $h, x \in \RR$, $t\in [0,\infty)$ satisfying $\lvert h \rvert \ll 1$ it holds that
\begin{equation}\label{2-d-dif}
\begin{split}
\frac{\partial^{2} f}{\partial x^{2}}(x,t) & = \frac{f(x+h,t)-2f(x,t)+f(x-h,t)}{h^2} + \mathcal{O}\bigl(\lvert h \rvert^\alpha \bigr)\\
& \approx \frac{f(x+h,t)-f(x,t)+f(x-h,t)-f(x,t)}{h^2}.
\end{split}
\end{equation}
Next, the Carath{\'e}odory theorem and \cref{2-d-dif} ensure the existence of $\psi_1 \colon \RR \times [0,\infty) \to \RR$, $\psi_{i,h} \colon \RR \times [0,\infty) \to \RR$, $i\in\{1,2\}$, $h\in\RR$, such that for all $h, x \in \RR$, $t\in [0,\infty)$ with $ \lvert h \rvert \ll 1$ it holds that
\begin{equation}\label{Car-1}
f(x+h,t) - f(x,t) = h \bigl[ \psi_{1,h}(x,t) \bigr] \quad \text{and} \quad f(x,t) - f(x-h,t) = h \bigl[ \psi_{1,-h}(x,t) \bigr]
\end{equation}
and
\begin{equation}\label{Car-1a}
\psi_{1,h}(x,t)-\psi_1(x,t) = h \bigl[ \psi_{2,h}(x,t) \bigr] \quad \text{and} \quad \psi_1(x,t)-\psi_{1,-h}(x,t) = h \bigl[\psi_{2,-h}(x,t) \bigr]. 
\end{equation}
This and the Carath{\'e}odory theorem further ensure that there exists $\psi_2 \colon \RR \times [0,\infty) \to \RR$ such that it holds for all $x \in \RR$, $t\in [0,\infty)$ that
\begin{equation}\label{Car-2}
\frac{\partial f}{\partial x}(x,t) = \psi_1(x,t) = \lim_{h\to 0} \psi_{1,h}(x,t) = \lim_{h\to 0} \psi_{1,-h}(x,t)
\end{equation}
and 
\begin{equation}\label{Car-3}
\frac{\partial^2 f}{\partial x^2}(x,t) = \frac{\partial \psi_1}{\partial x}(x,t) = \psi_2(x,t) = \lim_{h\to 0} \bigl( \psi_{2,h}(x,t) + \psi_{2,-h}(x,t) \bigr).
\end{equation}
This proves that for all $h,x\in \RR$, $t\in[0,\infty)$ such that $\lvert h \rvert \ll 1$ it holds that
\begin{equation}\label{new_gather1}
\psi_{2,h}(x,t) \approx \frac{1}{2} \frac{\partial^{2} f}{\partial x^{2}}(x,t)
\qquad \text{and} \qquad
\psi_{2,-h}(x,t) \approx \frac{1}{2} \frac{\partial^{2} f}{\partial x^{2}}(x,t)
\end{equation}
We now recapitulate the above arguments as an extension of the classical Carath{\'e}odory theorem (compare with, e.g., \cite[Theorem 6.1.5]{Bartle}).
\begin{theorem}\label{Carath-Gen}
Let $c \in \RR$, $t \in [0,\infty)$ and let $f \colon \RR \times [0,\infty) \to \RR$ be a function. Then $f$ is twice differentiable with respect to the variable $x$ at the point $(c,t)$ if and only if there exist $\psi_i \colon \RR \times [0,\infty) \to \RR$, $i\in\{1,2\}$, such that it holds that
\begin{equation}
f(x,t)-f(c,t)=\psi_1(x,t)(x-c)+\psi_2(x,t)(x-c)^2.
\end{equation} 
Moreover, in this case it holds that
\begin{equation}
\lim_{x\to c} \psi_1(x,t)= \frac{\partial f}{\partial x}(c,t)
\qquad \text{and} \qquad 
\lim_{x\to c} \psi_2(x,t) = \frac{1}{2} \frac{\partial^{2} f}{\partial x^{2}}(c,t).
\end{equation}
\end{theorem}

\begin{proof}[Proof of \cref{Carath-Gen}]
We omit the details of this proof for brevity. However, the result follows from arguments analogous to those used to generate \cref{Car-1}, \cref{Car-1a}, \cref{Car-2}, \cref{Car-3}, and \cref{new_gather1}.
The proof of \cref{Carath-Gen} is thus completed.
\end{proof}

The alternative derivation of \cref{eq:1} follows directly from \cref{Carath-Gen} (combined with the previous results outlined in \cref{derivation1}).
It is worth noting that the alternative derivation presented above only requires the function $f$ to be twice differentiable with respect to $x$ and once differentiable with respect to $t$. In fact, \cref{Carath-Gen} can be generalized to situations where $f$ possesses even less regularity. However, we leave these considerations for future endeavors.

\section{Closed-form solutions to the Einstein equation in application to GUI classification via retrieval time} \label{closed_form}

As should be clear from \cref{einstein-paradigm}, we intend to use the Einstein model of random jumps to interpret the results of experimental observations.
This approach should present a stark contrast to the traditional approach of employing Fick's law, flux conservation laws, and the thermodynamical law that density is proportional to the mass concentration function (see, for example, the classical work by \cite{Landau-1}). The arguments in this section motivate the novel approach with the particular case of GUI classification via retrieval time in mind.

\subsection{Development of closed-form solutions to the Einstein equation}

To accomplish our task, we assume that the number of molecules which form a compound of interest is proportional to the molecule's mass and that this molecular mass can be adequately represented by a scalar-valued function, say $C \colon \RR \times [0,\infty) \to \RR$.
In our particular case, we are interested in the setting where the domain of the process can be modeled via a one-dimensional tube. 
Moreover, we assume that there are $n\in\NN$ molecules of interest. For each $i\in\{1,2,\dots,n\}$ we assume that the associated molecular weight functions, $C_i \colon \RR \times [0,\infty) \to \RR$, satisfy \cref{eq:1} (with $f(x,t) \leftarrow C_i(x,t)$ for each $i\in\{1,2,\dots,n\}$ in the notation of \cref{eq:1}).
Next, we assume that for all $i\in\{1,2,\dots,n\}$ it holds that $D_i$, $\omega_i$, and $\gamma_i$ are constant. Furthermore, for each $i\in\{1,2,\dots,n\}$ we have the associated Green's function which satisfies for all $x \in \RR$, $t\in(0,\infty)$ that
\begin{equation}\label{alfa}
C_i(x,t) = \bigl(4\pi D_i\bigr)^{-\nicefrac{1}{2}} \exp\Bigl( -\tfrac{\gamma_i}{2 D_i} x + \bigl(\omega_i - \tfrac{\gamma_i^2}{4 D_i}\bigr) t - \tfrac{1}{2}\ln(t) - \tfrac{\lvert x\rvert^2}{4D_i t} \Bigr).
\end{equation}
This and straightforward calculus ensure that for all $i\in\{1,2,\dots,n\}$, $x \in \RR$, $t\in(0,\infty)$ it holds that
\begin{equation}\label{dgammaomega}
\left(\frac{\partial }{\partial t}-D_i\frac{\partial^2}{\partial x^2}+\gamma_i \frac{\partial }{\partial x}-\omega_i\right) C_i(x,t) = 0.
\end{equation}
Note that we can interpret the $C_i(x,t),$ $i\in\{1,2,\dots,n\}$, as analytic representations of the concentration function of the molecules which are injected into a sample of N-glycans (which are subsequently transported through the tube for classification via a mass spectrometer). 

Next assume that $L \in (0,\infty)$ is the point of observation (signal recording), the tube has infinite length (that is, we assume the domain of the problem to be $\RR$), and that for all $i\in\{1,2,\dots,n\}$ the initial concentration is modeled by $\delta_0(x)$ (the standard delta function).
This and \cref{alfa} demonstrate that for all $i\in\{1,2,\dots,n\}$, $t\in[0,\infty)$ it holds that
\begin{equation}\label{CL}
C_i(L,t) =
\begin{cases}
\bigl(4\pi D_i\bigr)^{-\nicefrac{1}{2}} \exp\Bigl( -\tfrac{\gamma_i}{2 D_i} L + \bigl(\omega_i - \tfrac{\gamma_i^2}{4 D_i}\bigr) t - \tfrac{1}{2}\ln(t) - \tfrac{L^2}{4D_i t} \Bigr) & \colon t > 0 \\
0 & \colon t = 0
\end{cases}.
\end{equation}

\subsection{Employing closed-form solutions to the Einstein equation to classify GUI via retrieval time}

Recall that we have experimental data from trials employing pure samples of each GUI. Thus, we will use this data to determine explicitly the $D_i$, $i\in\{1,2,\dots,n\}$, coefficients for all components of the N-glycans and the GUIs. 
In order to find the extreme values at the point of observation, we differentiate \cref{CL} to obtain for each $i\in\{1,2,\dots,n\}$, $t\in(0,\infty)$ that
\begin{equation}\label{extrim_crit}
\frac{\frac{\partial C_i}{\partial t}(L,t)}{C_i(L,t)} = \omega_i -\frac{\gamma_i^2}{4D_i}+\frac{L ^2}{4D_i t^2}-\frac{1}{2 t}
\qquad \text{and} \qquad
\frac{\frac{\partial^2 C_i}{\partial t^2}(L,t)}{C_i(L,t)} = -\frac{L^2}{2D_i t^3} + \frac{1}{2t^2}.
\end{equation}
Combining this with the assumption that $\Delta_e \approx 0$ (no drift --- cf.~\cref{remark1}) and the assumption that for all $i\in\{1,2,\dots,n\}$ it holds that $\omega_i=0$ (no absorption) yields that for all $i\in\{1,2,\dots,n\}$ it holds that
\begin{equation}\label{Dvst_max} 
D_i=\frac{1}{2}\left(\frac{L^2}{t_{\max}^i}\right),
\end{equation}
where for each $i\in\{1,2,\dots,n\}$ it holds that $t_{\max}^i \in (0,\infty)$ is the maximal critical point from \cref{extrim_crit}.
Combining \cref{Dvst_max} with experimental data will allow for the determination of the $D_i$, $i\in\{1,2,\dots,n\}$, by letting the retrieval time for each pure sample correspond to the $t_{\max}^i$, $i \in \{1,2,\dots,n\}$.

Note that \cref{fig: Retriv_Time_Pure} presents a graph of the actual spectrometer signals, with associated peaks, which correspond to the molecules of GUI passing through the receiver of the mass spectrometer (obtained via the methods outlined in \cref{GUI-structure}). As we can see, there are eight distinct GUIs and they each have different retrieval times. \cref{tab: T_i} presents the calculated diffusion coefficients $D_i$, $i\in\{1,2,\dots,n\}$, for each retention time. \cref{tab: T_i} and \cref{fig:D-vs-Tret_Claen} together demonstrate the expected inverse proportional relationship between retention time and diffusivity.

\begin{figure}[!htb]
\centering
\includegraphics[scale=0.8]{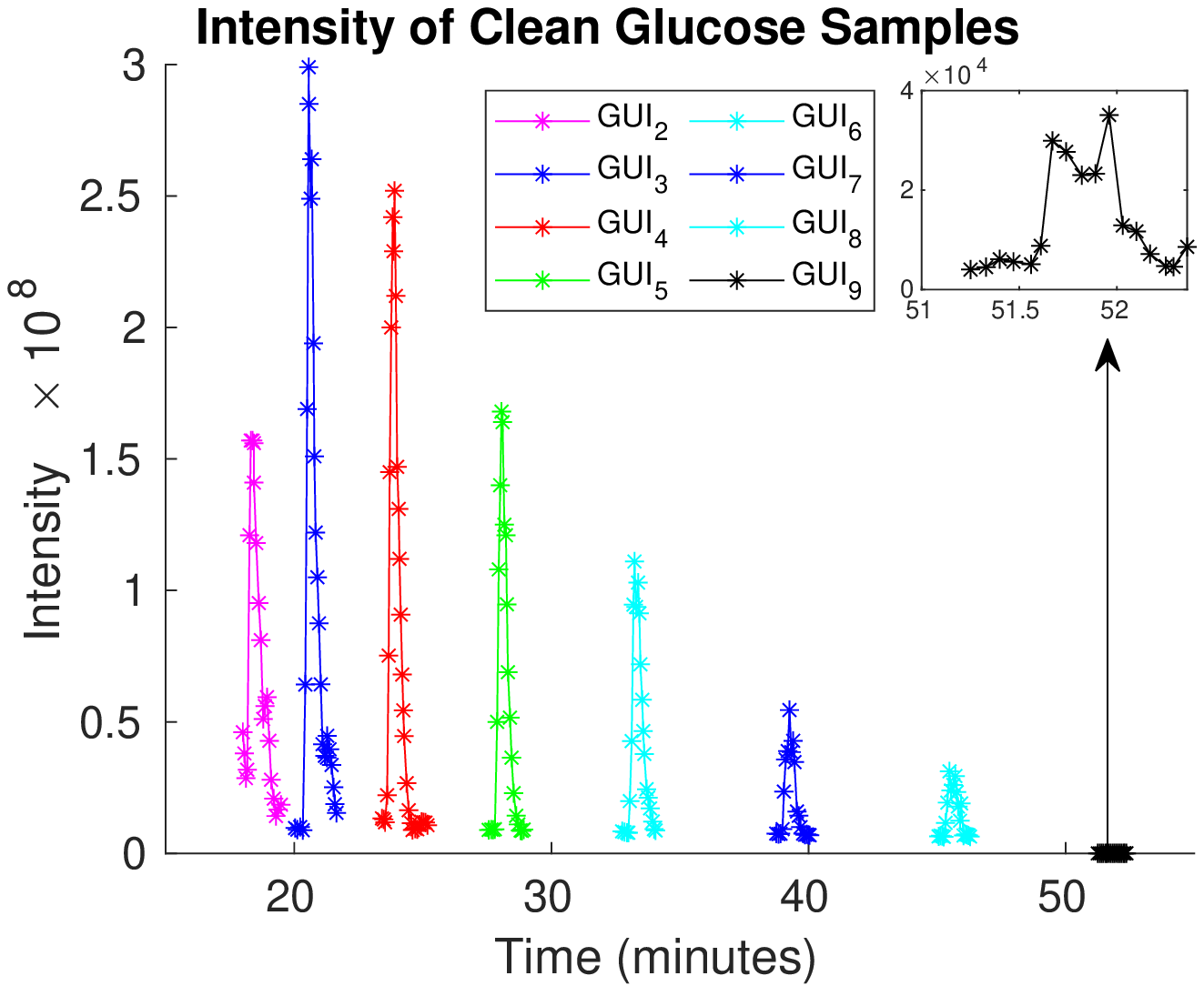}
\caption{Raw data from the mass spectrometer readings for pure glucose samples at the observation point $x=L$. Note how the peaks do not overlap. \label{fig: Retriv_Time_Pure}}
\end{figure}

The intuition provided by the Einstein paradigm tells us that the molecule's retrieval time should depend on the molecule's time of free travel and that the value should increase proportionally with the molecules mass. This observation is clearly supported by our data (see \cref{fig_4}), as we have the mass of the GUI molecules increases with their index value.
Furthermore, our post-processing --- which can be performed using our analytical formula for pure samples --- confirms this intuitive observation and Einstein's theory that the diffusion coefficients are inversely proportional to $\tau$ (cf.~\cref{assump2}).

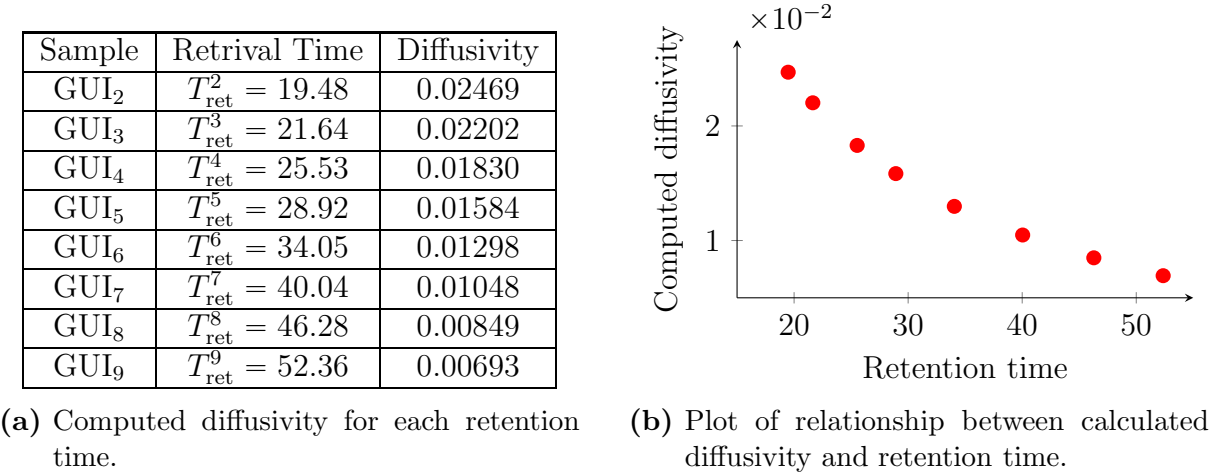
\begin{figure}[!htb]
\centering
\subcaptionbox{Computed diffusivity for each retention time.\label{tab: T_i}}[0.48\textwidth]{
\centering
\begin{tabular}{| c| c| c|}
\hline
Sample & Retrival Time & Diffusivity\\
\hline
$\gui_2$ & $T^2_{\text{ret}}=19.48$ & $
0.02469$\\
\hline
$\gui_3$ & $T^3_{\text{ret}}=21.64$ & $0.02202$\\
\hline
$\gui_4$ & $T^4_{\text{ret}}=25.53$ & $0.01830$ \\
\hline
$\gui_5$ & $T^5_{\text{ret}}=28.92$ &$0.01584$\\
\hline
$\gui_6$ & $T^6_{\text{ret}}=34.05$ & $0.01298$ \\
\hline
$\gui_7$ & $T^7_{\text{ret}}=40.04$ &$
0.01048$\\
\hline
$\gui_8$ & $T^8_{\text{ret}}=46.28$ & $
0.00849$\\
\hline
$\gui_9$ & $T^9_{\text{ret}}=52.36$ &$0.00693$\\
\hline
\end{tabular}
}~~~~~
\subcaptionbox{Plot of relationship between calculated diffusivity and retention time. \label{fig:D-vs-Tret_Claen}}[0.48\textwidth]{
\centering
\begin{tikzpicture}
\begin{axis}[
    height = 5cm,
    width = 4cm,
    clip = false,
    tick scale binop=\times,
    domain=15:55,
    xmin=15, xmax=55,
    ymin=0.005, ymax=0.0275,
    x=1.5mm,
    axis lines=left,
    xlabel near ticks,
    xlabel={Retention time},
    ylabel={Computed diffusivity}
    ]
\node[circle,fill,red,inner sep=2pt] at (axis cs:19.48,0.02469) {} ;
\node[circle,fill,red,inner sep=2pt] at (axis cs:21.64,0.02202) {} ;
\node[circle,fill,red,inner sep=2pt] at (axis cs:25.53,0.01830) {} ;
\node[circle,fill,red,inner sep=2pt] at (axis cs:28.92,0.01584) {} ;
\node[circle,fill,red,inner sep=2pt] at (axis cs:34.05,0.01298) {} ;
\node[circle,fill,red,inner sep=2pt] at (axis cs:40.04,0.01048) {} ;
\node[circle,fill,red,inner sep=2pt] at (axis cs:46.28,0.00849) {} ;
\node[circle,fill,red,inner sep=2pt] at (axis cs:52.36,0.00693) {} ;
\end{axis}
\end{tikzpicture}
}
\caption{A depiction of the relationship between retention time and calculated diffusivity. \cref{tab: T_i} and \cref{fig:D-vs-Tret_Claen} make it clear that there is a nonlinear inverse proportional relationship between diffusivity and retention time.\label{fig_4}}
\end{figure}

Indeed, consider two GUIs, say, $\gui_1$ and $\gui_2$. These two GUIs correspond to the molecules $M_1$ and $M_2$, respectively. 
Since these are distinct types of glucose molecules, it is well-known that the molecules will not develop novel chemical bonds throughout the experiment.
Therefore, their ``free jump'' lengths are mutually independent. These ``free jumps'' of molecule $M_1$ corresponds to $\tau_1$ and the ``free jumps'' of molecule $M_2$ corresponds to $\tau_2$ (cf.~\cref{assump2}).
If we let $C_1(x,t)$ and $C_2(x,t)$ denote the concentrations at position $x$ and time $t$ of the molecules $M_1$ and $M_2$, respectively, then both of these functions will satisfy \cref{eq:1} (each with their appropriate associated parameters).
This intuition combined with Einstein's paradigm and the assumptions outlined above then implies for all $i\in\{1,2,\dots,n\}$ it holds that
\begin{equation}\label{D_i}
D_i=\frac{1}{2\tau_i}\int_{-\infty}^\infty (\Delta-\Delta_e)^2\varphi(\Delta)\,d\Delta
\qquad\text{and}\qquad
\gamma_i=\frac{1}{\tau_i}\int_{-\infty}^\infty \Delta\varphi(\Delta)\,d\Delta,
\end{equation} 
where $\varphi(\Delta)$ is frequency at which free jumps of the length $\Delta$ occur.
Note that in this formulation we have assumed that the $C_i$, $i\in\{1,2,\dots,n\}$, are associated to each molecule $M_i$, $i\in\{1,2,\dots,n\}$, and no novel molecular bonds are formed.
Therefore, the $\Delta_i$, $i\in\{1,2,\dots,n\}$, are the lengths of the associated ``free jumps'' of the entire compound of molecules of type $M_i$, $i\in\{1,2,\dots,n\}$. We also assumed that the expected length of free jumps were the same for all $i\in\{1,2,\dots,n\}$; that is, we assumed that $\Delta_{i,e} = \Delta_e$, $i\in\{1,2,\dots,n\}$ (cf.~\cref{remark1}).
Moreover, the molecules are arranged by their molecular weight as their lengths are proportional to $i \in \{1,2,\dots,n\}$. Clearly, a higher molecular mass is associated to smaller ``free jump'' lengths. Mathematically, this means for all $i,j\in\{1,2,\dots,n\}$ with $i < j$ it holds that $\tau_i < \tau_j$.

Note that the experimental observations presented in \cref{fig_4} support precisely this claim. Further observe that the velocity of the filtration (drift) is so small that diffusion is the dominant transport property.
Indeed, if drift had a larger influence than diffusion, then the decrease of the retrial time with respect to the GUI mass would be linear. Since \cref{fig:D-vs-Tret_Claen} clearly indicates an inversely propositional relationship, we conclude that diffusion is dominant. This is an important observation as drift mainly depends upon the boundary conditions of a problem whereas diffusion depends only on object versus tube structure properties. This crucial observation is what allows for the proposed classification method to work so well.
Finally, we mention that in this preliminary result we have also ignored the associated absorption rates. However, classification methods for complex serum samples will consider these effects.


\begin{remark}\label{remark_5}
From the Einstein paradigm it follows that the calculated diffusion coefficients are inversely proportional to $\tau$ (the ``free jump times'' --- cf.~\cref{assump2}) of each molecule. \cref{Dvst_max} provides an explicit relationship between retrieval time and the ``free jumps'' for molecule $M_i$, $i\in\{1,2,\dots,n\}$: smaller retrieval times correspond to larger $\tau$. 
\end{remark}

Note that further refinement of the classification criteria may come from information which is hidden in the dynamics of the intensity of the signal prior to the retrieval time. In our future research, we intend to generalize our procedure by further incorporating the area under the graph of the signal, prior to retrieval time, for each molecule. Such considerations result in a need for better understanding the following functional   
\begin{equation}\label{Class-Index}
I(t) = \frac{d}{dt} \ln \left(\int_{0}^{t} C(L,\tau) \, d\tau\right) = \frac{C(L,t)}{\int_{0}^{t} C(L,\tau) \, d\tau}.
\end{equation}
From this it is clear that if $T_{\text{ret}}$ is the retrieval time then \cref{Class-Index} evaluates to $(T_{\text{ret}})^{-1}$, {\em if} we employ a (very) rough numerical approximation of the integral which employs one rectangle of height $H = T_{\text{ret}}$. This of course agrees with \cref{Dvst_max}.  
In this regard, one can see that \cref{Class-Index} is a true generalization of \cref{Dvst_max}.


\section{Serum classification using GUIs as markers}\label{classification_sec}

We will use the general ideas of \cref{remark_5} as a basis for the development of our classification algorithm using GUIs as markers.
The basic idea of the algorithm (which was outlined in the \bref{proposed_alg}{Proposed Algorithm}) consists of determining the 
diffusivity coefficients of ``marker'' molecules in a base sample of interest (cf.\ \cref{eq:1}). The remaining molecules are classified by grouping them based on diffusivity coefficient ranges and computing the associated absorption coefficients, which serve as a correction term of sorts (cf.\ \cref{eq:1}).
We assume throughout that drift coefficients are negligible (cf.\ \cref{eq:1}). This assumption is justified due to the intended use of post-processing in the actual physical experiment.

\subsection{Serum classification algorithm}

The proposed algorithm for the classification of a given unknown sample using two parameters, diffusivity and absorption, consists of five main steps. Note that throughout this section the primary focus in the classification of N-glycans.

\medskip

\begin{jlp_alg*}\hypertarget{alg1}{}
Let $N \in \NN$ and consider the set of samples $\mathcal{A} = \{A_0,A_1,\dots,A_N\}$.
\begin{enumerate}[label={\bf Step \arabic*},labelindent=10pt,leftmargin=!,
labelwidth=\widthof{\ref{last_item}},
ref={Step \arabic*}]
\item
\label{step1}
Select a base sample, which without loss of generality we assume to be $A_0$. Let $n_0 \in \NN$ represent the number of molecules in $A_0$. Inject the ``marker molecules'', or GUIs --- which we designate as $\gui_1, \gui_2, \dots ,\gui_8$, into the sample $A_0$. 
Without loss of generality, assume that the GUIs are indexed with respect to increasing mass.
Collect all retention times $T_i^0$, $i\in\{1,2,\dots,n_0\}$, from the mass spectrometer. Identify (manually) which retention times are associated to the GUIs.
\item 
\label{step2}
Using the retention times from \ref{step1}, calculate the diffusion coefficients for each GUI, which we designate as $D_i^0$, $i\in\{1,2,\dots,8\}$, via \cref{Dvst_max}.
Set the associated absorption coefficients, which we designate as $\omega_i^0$, $i\in\{1,2,\dots,8\}$, to zero.
This completes the baseline classification procedure (if desired, one can classify the remaining objects in $A_0$).
\item
\label{step3}
Take a new (unknown) sample, which without loss of generality we assume to be $A_1$, and again inject the GUI molecules from \ref{step1}. 
Let $n_1 \in \NN$ represent the number of molecules in $A_1$ (after injection with GUI molecules).
Pass $A_1$ through the mass spectrometer and again collect all retention times, $T^1_i$, $i\in\{1,2,\dots,n_1\}$.
Using \cref{Dvst_max} compute all associated diffusivity coefficients, $D_i^1$, $i\in\{1,2,\dots,n_1\}$.
\item
\label{step4}
Using the results from \ref{step1} find the coefficients $i_1,i_2,\dots,i_n \in \{1,2,\dots,n_1\}$ which for each $j\in\{1,2,\dots,8\}$ satisfy that
\begin{equation}\label{eq:30}
\{1,2,\dots,n_1\} \ni i_j = \min\bigl\{k \in \{1,2,\dots,n_1\} \backslash \{i_1,i_2,\dots,i_{j-1}\} \colon \lvert D_{i_j}^1 - D_j^0 \rvert \bigr\}.
\end{equation}
Note that \cref{eq:30} is well-defined for all $j\in\{1,2,\dots,8\}$ by construction.
The molecules associated with the indices $i_1,i_2,\dots,i_8$ are the GUI markers $\gui_1, \allowbreak \gui_2,\allowbreak \dots, \allowbreak \gui_8$.
Let $\mathbb{D}^1, \mathbb{T}^1 \in \RR^8$ satisfy that
\begin{equation}\label{eq:31}
\mathbb{D}^1 = \{D_{i_1}^1, D_{i_2}^1, \dots, D_{i_8}^1\} \qquad \text{and} \qquad \mathbb{T}^1 = \{T_{i_1}^1, T_{i_2}^1, \dots, T_{i_8}^1\}.
\end{equation}
Finally, let $\omega_{i_j}^1$, $j\in\{1,2,\dots,8\}$, satisfy for all $j \in \{1,2,\dots,8\}$ that $\omega_{i_j}^1 = 0$.
\item 
\label{step5}
We now classify the remaining molecules from $A_1$ using \cref{eq:31} to ``classify'' all remaining objects.
Note that it is the case that 
$\max_{i \in \{1,2,\dots,n_1\}} T_i^1 \in \mathbb{T}^1$.
For all $k\in\{1,2,\dots,n_1\} \backslash \{i_1,i_2,\dots,i_8\}$ we compute the associated $\omega_k^1$ as
\begin{equation}
\omega_i^1 = 
\begin{cases}
\bigl(T_k^1\bigr)^{-1}\left[1 - \bigl(2 T_k^1 D_{i_1}^1\bigr)^{-1}\right] & \colon 0 < T_k^1 < T_{i_1}^1 \\
\bigl(T_k^1\bigr)^{-1}\left[1 - \bigl(2 T_k^1 D_{i_j + 1}^1\bigr)^{-1}\right] & \colon T_{i_j}^1 < T_k^1 < T_{i_{j+1}}^1,\ j\in\{1,2,\dots,7\}
\end{cases}
\end{equation}
(cf.~\cref{extrim_crit}).
\item
\label{last_item}
Repeat \ref{step3}, \ref{step4}, and \ref{step5} for the remaining $A_i \in \mathcal{A}$, $i\in\{2,3,\dots,N\}$.
\end{enumerate}
\end{jlp_alg*}

\medskip

It is worth noting that only the diffusivity coefficients calculated from \ref{step1} and \ref{step2} need to be stored for future use. This data set serves as the baseline learning procedure for the data-driven classification algorithm. Thus, while the initial manual classification can be tedious, it results in an algorithm which can be used to classify a large number of other unknown samples (of appropriate type).

\begin{remark}
\bref{alg1}{Main Algorithm} can be significantly improved if we consider the Classification Index to be a {\em time series} instead. 
In this case, we may consider \cref{Class-Index} as the basis for classification. When considering a serum containing only the ``marker'' molecules, it follows that $t$ is the retrieval time in \cref{Class-Index}. An analogous (but improved) algorithm can then be obtained by approximating the integral implicitly (for numerical stability) and employing a Newton-Raphson-type method to solve for the desired parameters. As noted earlier, this will be a focus of forthcoming work.
\end{remark}  

\subsection{Description of the experimental set up and materials used}\label{GUI-structure}

In this section we briefly outline the materials and experimental protocols followed in order to obtain our experimental data, which is used for comparison. 

\subsubsection{Material}

Standard glycoproteins, fetuin and ribonuclease B (RNase B) and pooled human blood serum (HBS) were purchased from Sigma Aldrich (St. Louis, MO). Formic acid (FA), borane-ammonia, dimethyl sulfoxide (DMSO), iodomethane and, sodium hydroxide beads were also obtained from the same vendor. HPLC grade water was obtained from Avantor Performance Materials (Center Valley, PA). HPLC grade acetonitrile (ACN), methanol, and ethyl alcohol were supplied by Fisher Scientific (Fair Lawn, NJ). PNGase F enzyme and 10XG7 buffer (0.5M phosphate buffer saline) were purchased from New England Biolabs.  

\subsubsection{Sample preparation}

Model Glycoproteins and Dextrin: 20 $\mu$g each of fetuin and RNase B were mixed with G7 buffer to get a final concentration of 20 mM for the buffer. The samples were then denatured at 90$^\circ$C for 30 minutes. Samples were then cooled at room temperature and treated with 1.0 $\mu$l of PNGase F, followed by incubation at 37$^\circ$C for 18 hours. PNGase F digestion was followed by precipitation of de-N-glycosylated proteins with 90\% ethanol at $-20^\circ$C. Reduction of reducing ends of the purified glycans was done by addition of 10 $\mu$ of borane-ammonia complex (10 $\mu g$/$\mu$L) and incubating it at 60$^\circ$C water bath for one hour. Methanol was later used to remove borane in the form of borate from reduced glycan samples. The methanol washing step was repeated three times to ensure the complete removal of borate from the samples. Reduction was then followed by permethylation of the samples, using a previously reported method \cite{15}. For this purpose, reduced and dried glycan samples were resuspended in 1.2 $\mu$L and 30 $\mu$L of DMSO. Later, 20 $\mu$L of iodomethane was added to the samples and they were loaded on DMSO soaked sodium hydroxide beads packed in spin columns. The spin columns were washed with 200 $\mu$L of DMSO, using a centrifuge at 1800 rpm for two minutes, prior to the loading of samples. 

Once loaded, the samples were incubated at room temperature for two minutes. After 25 minutes, an additional 20 $\mu$L of iodomethane was added and the samples were again incubated at the room temperature for 15 minutes. Permethylated glycans were then collected by centrifugation at 1800 rpm for two minutes. For complete elution of permethylated glycans, 30 $\mu$L of ACN was added to the spin columns and again the eluants were collected by centrifugation. Permethylated glycans were further dried and resuspended in 20\% ACN and 0.1\% FA. Each of the samples were run in triplicates and 1 $\mu$g of the samples were injected for each run. 
Dextrin standard was mixed with the samples prior to reduction and, therefore, was reduced and permethylated with each sample. 1 $\mu g$ of sample was spiked with 100 ng of dextrin.

\subsubsection{Human blood serum} 

10 $\mu$L of human blood serum was mixed with 90 $\mu$L of G7 buffer to get a final concentration of 20mM for the buffer. Proteins from the samples were denatured in 90$^\circ$C water bath for 30 minutes. After cooling at room temperature, 1.2 $\mu$L of PNGase F was added to the samples. They were then incubated at 37$^\circ$C for 18 hours. After the completion of the incubation, proteins were precipitated at $-20^\circ$C for one hour. Reduction and permethylation were then performed as described previously for model glycoproteins. Resuspension was again done in 20\% ACN and 0.1\% FA. 1 $\mu$L of the serum samples were then injected for each of the triplicate runs.    

\subsubsection{Liquid chromatography (LC) conditions}   

Chromatography was performed on UltiMate 3000 Nano UHPLC system using C18 column. Optimum temperature for the oven was kept at 55$^\circ$C. A solution of 98\% water, 0.2\% ACN, and 0.1\% FA was utilized as mobile phase A while, 100\% ACN and 0.1\% FA was mobile phase B. Initially, the gradient was set at 20\% mobile phase B. It was then increased to 42\% in 11 minutes. After 48 minutes, it was increased to 55\% and then changed to 90\% at 49 minutes. It remained at 90\% for 54 minutes of  total sample run and plummeted to 20\% again for equilibration of the column for the final six minutes.

\subsubsection{Mass spectrometry (MS) conditions}  

LTQ Orbitrap Velos (Thermo Scientific) was used to analyze the samples. The mass spectrometer was set to the positive ion mode with an ESI voltage of 1.6 kV. Full MS was performed at 100,000 resolution with 200-2000 m/z scan range.  MS2 was acquired with collision induced dissociation (CID) and higher energy collision dissociation (HCD) with normalized dissociation energies of 30\% and 45\%, respectively. Activation Q (one of the parameters used in Mass spectrometry methods --- the value controls the radio frequency applied to control fragmentation of ions during analysis) was 0.25. Injection time was 10 ms. Repeat count of dynamic exclusion and repeat duration were 2 s and 30 s, respectively. The exclusion duration was 60 s. The four most intense ions were selected from the full MS for further CID and HCD based dissociation by applying data-dependent acquisition mode. The precursor ion selection window was 1.50. The MS2 intensity threshold was 5000 counts. Singly charged ions were excluded for MS2.

\subsubsection{Data analysis} 

The extracted ion chromatograms (EIC) of full MS data were used to determine the glycan composition as well as retention times of reduced and permethylated glycans derived from model glycoproteins, and human blood serum, with a mass tolerance of 10 ppm. Retention times of reduced and permethylated glucose units were also determined using the EIC. 

\subsection{Data classification of an actual experiment using the proposed algorithm}\label{experiment-and classification}

We now demonstrate \bref{alg1}{Main Algorithm} and its efficacy through an experimental example. We will use data obtained via the methods and procedures outlined in \cref{GUI-structure}.
To obtain the data in \cref{table2}, \cref{table3}, and \cref{table4} below we implemented our algorithm for a particular set of experiments (obtained via the methods outlined in \cref{GUI-structure}). In the first table (left table) in \cref{table2}, the GUIs were known (green rows). Then this data is then used according to \bref{alg1}{Main Algorithm} to identify the eight GUI markers in five other (unknown) experiments via the same instrumentation.
We then later used precise verification methods to determine that \bref{alg1}{Main Algorithm} can distinguish GUIs with an error of no more than two percent.  

\begin{figure}[!htb]
\centering
\scalebox{0.8}{
\begin{tabular}{| c| c| c|}
\hline
$T_{\text{ret}}$ & Diffusivity & Absorption \\
\hline
\rowcolor{green}$54.11$ & $0.009240$ & $0$\\
\hline
\rowcolor{green}$50.47$ & $0.009907$ & $0$\\
\hline
$47.19$ & $0.010595$ & $0.00088892$ \\
\hline
$43.77$ & $0.011423$ &$0.00013409$\\
\hline
$43.58$ & $0.011473$ & $0.00849000$ \\
\hline
\rowcolor{green}$43.26$ & $0.0115580$ &$0$\\
\hline
$41.87$ & $0.011942$ & $0.00186390$\\
\hline
$41.16$ & $0.012148$ &$0.00167795$\\
\hline
$40.82$ & $0.012249$ &$0.00158391$\\
\hline
$38.47$ & $0.012997$ &$0.00083303$\\ 
\hline
$37.49$ & $0.013337$ &$0.00045986$\\
\hline
$37.15$ & $0.013459$ &$0.00032089$\\
\hline
\rowcolor{green}$36.41$ & $0.013732$ &$0$\\
\hline
$35.66$ & $0.014021$ &$0.00301127$\\
\hline
$34.89$ & $0.014368$ &$0.00265176$\\
\hline
$33.14$ & $0.015088$ &$0.00183817$\\
\hline
\rowcolor{green}$30.33$ & $0.016485$ &$0$\\
\hline
$29.65$ & $0.016863$ &$0.00437691$\\
\hline
\rowcolor{green}$25.20$ & $0.019841$ &$0$\\
\hline
\rowcolor{green}$21.28$ & $0.023496$ &$0$\\
\hline
\rowcolor{green}$18.68$ & $0.026767$ &$0$ \\
\hline
\end{tabular}
\hspace{5mm}
\begin{tabular}{| c| c| c| c|}
\hline 
$T_{\text{ret}}$ & Diffusivity & $D_{\text{base}}$ & Absorption \\
\hline
\rowcolor{green}$54.11$ & $0.009240$ & $0.009240$  & $0$\\
\hline
\rowcolor{green}$50.48$ & $0.0099049$ &  $0.0099069$  &  $0$\\
\hline
$47.23$ & $0.0105865$ &   & $-0.001456959$ \\
\hline
$43.77$ & $0.0114233$ &   &$-0.00350243$\\
\hline
$43.65$ & $0.0114548$ &   & $-0.003584695$ \\
\hline
\rowcolor{green}$43.24$ & $0.0115634$ & $0.0115580$  &$0$\\
\hline
$41.95$ & $0.0119190$  &    & $-0.000966017$\\
\hline
$41.17$ & $0.0121448$  &    &$-0.001463154$\\
\hline
$40.85$ & $0.0122399$   &    &$-0.00167793$\\
\hline
$38.48$ & $0.0129938$  &    &$-0.003491567$\\ 
\hline
$37.51$ & $0.0133298$ &    &$-0.004363895$\\
\hline
$37.12$ & $0.0134698$  &    &$-0.004739115$\\
\hline
\rowcolor{green}$36.40$ & $0.0137363$ & $0.0137325$ &$0$\\
\hline
$35.66$ & $0.0140213$  &    &$-0.001148128$\\
\hline
$34.85$ & $0.0143472$  &    &$-0.001869047$\\
\hline
$33.1$ & $0.0151057$  &    &$-0.003669189$\\
\hline
\rowcolor{green}$30.31$ & $0.0164962$ & $0.016485$  & $0$\\
\hline
$29.64$ & $0.0168691$ &  &$-0.003938399$\\
\hline
\rowcolor{green}$25.21$ & $0.0198334$  & $0.0198413$ &$0$\\
\hline
\rowcolor{green}$21.30$ & $0.0234742$ & $0.023496$  &$0$\\
\hline
\rowcolor{green}$18.70$ & $0.0267380$ & $0.026767$ & $0$\\
\hline
\end{tabular}
}
\caption{[LEFT] Table containing experimental data for base sample, $A_0$, and associated calculated values. [RIGHT] Table containing experimental data for second experiment, $A_1$, with associated calculated values.}\label{table2}
\end{figure}

\begin{figure}[!htb]
\centering
\scalebox{0.8}{
\begin{tabular}{| c| c| c| c|}
\hline 
$T_{\text{ret}}$ & Diffusivity & $D_{\text{base}}$ & Absorption \\
\hline
\rowcolor{green}$54.11$ & $0.0092404$ & $0.009240$  & $0$\\
\hline
\rowcolor{green}$50.49$ & $0.0099030$ &  $0.0099069$  &  $0$\\
\hline
$47.25$ & $0.0105820$ &   & $-0.001451247$ \\
\hline
$43.72$ & $0.0114364$ &   &$-0.003541835$\\
\hline
$43.67$ & $0.0114495$ &   & $-0.003576169$ \\
\hline
\rowcolor{green}$43.32$ & $0.0115420$ & $0.0115580$  &$0$\\
\hline
$41.92$ & $0.0119275$  &    & $-0.000796683$\\
\hline
$41.21$ & $0.0121330$  &    &$-0.001242445$\\
\hline
$40.88$ & $0.0122309$   &   &$-0.001460051$\\
\hline
$38.54$ & $0.0129735$  &   &$-0.003218137$\\ 
\hline
$37.56$ & $0.0133120$ &   &$-0.004082924$\\
\hline
$37.23$ & $0.0134300$  &    &$-0.004393708$\\
\hline
\rowcolor{green}$36.47$ & $0.0137099$ & $0.0137325$ &$0$\\
\hline
$35.72$ & $0.0139978$  &    &$-0.000587812$\\
\hline
$34.84$ & $0.0143513$  &    &$-0.001342862$\\
\hline
$33.12$ & $0.0150966$  &   &$-0.003053966$\\
\hline
\rowcolor{green}$30.33$ & $0.0164853$ & $0.016485$  & $0$\\
\hline
$29.64$ & $0.0168691$ & &$-0.000785403$\\
\hline
\rowcolor{green}$25.21$ & $0.0198334$  & $0.0198413$ &$0$\\
\hline
\rowcolor{green}$21.32$ & $0.0234522$ & $0.023496$  &$0$\\
\hline
\rowcolor{green}$18.68$ & $0.0267666$ & $0.026767$ &$0$\\
\hline
\end{tabular}
\hspace{5mm}
\begin{tabular}{| c| c| c| c|}
\hline 
$T_{\text{ret}}$ & Diffusivity & $D_{\text{base}}$ & Absorption \\
\hline
\rowcolor{green}$54.14$ & $0.0092353$ & $0.009240$  & $0$\\
\hline
\rowcolor{green}$50.67$ & $0.0098678$ &  $0.0099069$  &  $0$\\
\hline
$46.71$ & $0.0107043$ &  & $-0.001814995$ \\
\hline
\rowcolor{green}$43.40$ & $0.0115207$ & $0.0115580$  &$0$\\
\hline
$41.17$ & $0.0121448$ &   & $-0.005604823$ \\
\hline
\rowcolor{green}$36.51$ & $0.0136949$ & $0.013732491$  &$0$\\
\hline
$36.15$ & $0.0138313$  &   & $-0.000275477$\\
\hline
$35.56$ & $0.0140607$  &   &$-0.000751277$\\
\hline
$30.73$ & $0.0162707$   &   &$-0.006120723$\\
\hline
\rowcolor{green}$30.37$ & $0.0164636$  & $0.016485328$   &$0$\\ 
\hline
$25.87$ & $0.0193274$ &   &$-0.006723875$\\
\hline
\rowcolor{green}$25.23$ & $0.0198177$  & $0.01984127$   &$0$\\
\hline
\rowcolor{green}$21.31$ & $0.0234632$ & $0.023496241$ &$0$\\
\hline
\rowcolor{green}$18.69$ & $0.0267523$  & $0.026766595$   &$0$\\
\hline
\end{tabular}
}
\caption{[LEFT] Table containing experimental data for third experiment, $A_2$, with associated calculated values. [RIGHT] Table containing experimental data for fourth experiment, $A_3$, with associated calculated values.}\label{table3}
\end{figure}

As noted above, the left table in \cref{table2} consists of the data obtained from mass spectrometer readings for the sample $A_0$ (cf.~\bref{alg1}{Main Algorithm}) and the calculated absorption coefficients for that sample. The first column of the table contains the retention time data. The second column contains the corresponding calculated diffusion coefficients. The last column contains the calculated absorption coefficients.  The green rows in each table represent the data corresponding to the injected GUI molecules. As mentioned above, for the sample
$A_0$, manual intervention to the mass spectrometer data is mandatory for finding the GUI retention times. See \cref{GUI-structure} for more details.

\begin{figure}[!htb]
\centering
\scalebox{0.8}{
\begin{tabular}{| c| c| c| c|}
\hline 
$T_{\text{ret}}$ & Diffusivity & $D_{\text{base}}$ & Absorption \\
\hline
\rowcolor{green}$54.15$ & $0.0092336$ & $0.009240$  & $0$\\
\hline
\rowcolor{green}$50.79$ & $0.0098445$ &  $0.0099069$  &  $0$\\
\hline
$46.83$ & $0.0106769$ &   & $-0.001805705$ \\
\hline
\rowcolor{green}$43.54$ & $0.0114837$ & $0.011558021$  &$0$\\
\hline
$41.31$ & $0.0121036$ &  & $-0.001306756$ \\
\hline
\rowcolor{green}$36.62$ & $0.0136537$ & $0.013732491$  &$0$\\
\hline
$36.27$ & $0.0137855$  &   & $-0.000266056$\\
\hline
$35.7$ & $0.0140056$  &   &$-0.000721857$\\
\hline
$30.85$ & $0.0162075$   &  &$-0.006062692$\\
\hline
\rowcolor{green}$30.46$ & $0.0164150$  & $0.016485328$   &$0$\\ 
\hline
$25.98$ & $0.0192456$ &    &$-0.006637426$\\
\hline
\rowcolor{green}$25.32$ & $0.0197472$  & $0.01984127$   &$0$\\
\hline
\rowcolor{green}$21.35$ & $0.0234192$ & $0.023496241$ &$0$\\
\hline
\rowcolor{green}$18.71$ & $0.0267237$  & $0.026766595$   &$0$\\
\hline
\end{tabular}
\hspace{5mm}
\begin{tabular}{| c| c| c| c|}
\hline 
$T_{\text{ret}}$ & Diffusivity & $D_{\text{base}}$ & Absorption \\
\hline
\rowcolor{green}$54.12$ & $0.009238729$ & $0.009240$  & $0$\\
\hline
\rowcolor{green}$50.77$ & $0.009848336$ &  $0.0099069$  &  $0$\\
\hline
$46.80$ & $0.010683761$ &  & $-0.001812587$ \\
\hline
\rowcolor{green}$43.46$ & $0.011504832$ & $0.011558021$  &$0$\\
\hline
$41.23$ & $0.012127092$ &  & $-0.001311832$ \\
\hline
\rowcolor{green}$36.56$ & $0.013676149$ & $0.013732491$  &$0$\\
\hline
$36.17$ & $0.013823611$  &   & $-0.000298104$\\
\hline
$35.56$ & $0.014025245$  &   &$-0.000716015$\\
\hline
$30.78$ & $0.016244314$   &   &$-0.006100854$\\
\hline
\rowcolor{green}$30.43$ & $0.016431153$  & $0.016485328$   &$0$\\ 
\hline
$25.99$ & $0.019238169$ &   &$-0.006573103$\\
\hline
\rowcolor{green}$25.28$ & $0.019778481$  & $0.01984127$   &$0$\\
\hline
\rowcolor{green}$21.33$ & $0.023441163$ & $0.023496241$ &$0$\\
\hline
\rowcolor{green}$18.71$ & $0.026723677$  & $0.026766595$   &$0$\\
\hline
\end{tabular}}
\caption{[LEFT] Table containing experimental data for fifth experiment, $A_4$, with associated calculated values. [RIGHT] Table containing experimental data for sixth experiment, $A_5$, with associated calculated values.} \label{table4}
\end{figure}

The second table in \cref{table2} contains the data that was obtained from the mass spectrometer after the experiment for sample $A_1$ (cf.~\bref{alg1}{Main Algorithm}) and implementation of \bref{alg1}{Main Algorithm} for that sample. The first column of the table contains the retention times for the sample. The second column contains the corresponding calculated diffusion coefficients. The third column contains the coefficients of the GUI molecules from sample $A_0$ (used for classification). The last column contains the calculated absorption coefficients. The tables in \cref{table3} and \cref{table4} may be interpreted in a similar manner as the second table of \cref{table2}.

\begin{remark}
As is well understood, manual classification of molecular structures via the mass spectrometer is a time consuming and difficult task. 
However, \bref{alg1}{Main Algorithm} is demonstrated to be able to successfully identify GUI molecules in all experimental samples and also classify the remaining molecular structures using only retention times (after the initial classification of the ``marker'' molecules.
\end{remark}


\section{Conclusions and future endeavors}\label{conclusions}

In this article we developed a novel mathematical method which allows for an efficient classification of experimental samples using a Glucose Unit Index as the reference frame. These interpretations were associated directly to the experimental spectrometer data in particular examples.
In order to develop the novel method, we presented a data-driven partial differential equation model based on modified Einstein paradigm arguments. This extends Einstein's original study of Brownian motion to the situation of a more general conservation law. Once this data-driven model is obtained, we develop a closed-form solution of the model in order to avoid numerical approximations. A simple learning procedure is performed in order to determine the solutions coefficients on an initial sample. These coefficients are then used in additional learning procedures for computing the coefficients associated with unknown samples.

In order to further justify our method, we provide physical interpretations of the model coefficients. These coefficients are shown to be related to the experimental retrieval times, as well. This serves as the basis for the novel algorithm used for data classification (cf.~\bref{alg1}{Main Algorithm}). 
Moreover, the proposed algorithm is successfully implemented and its efficacy is demonstrated via examples. Through the consideration of six independent samples we demonstrate that our method successfully classifies the unknown samples with errors which do not exceed two percent. Moreover, our novel method can be shown to be ten percent more accurate than the traditional method of using retrieval time only for classification.

It is important to mention that the retrieval time and shape of the peak in the spectrometer data depend on three main parameters: drift (velocity), variance (diffusion), and absorption. In this article we demonstrated that if the drift is fixed, due to reprocessing of the data, then the diffusion and absorption coefficients can accurately classify molecules in an unknown sample. Our future endeavors will focus on including all three characteristics for the N-glycan classification. This inclusion will make use of iterative deep learning algorithms of Kolmogorov-type (see, for instance, \cite{Kolomogorov-1,Tikhomirov,Yang}).

%

\bibliographystyle{acm}
\bibliography{data_driven_bib}

\end{document}